\documentclass[12pt,a4paper]{amsart}
\usepackage{amsmath}
\usepackage{amsthm}
\usepackage{amsfonts}
\usepackage{amssymb}
\usepackage{latexsym}
\usepackage{amscd}
\usepackage{stmaryrd}
\usepackage{amsbsy}
\usepackage{pb-diagram}
\usepackage{tikz}
\usepackage{fullpage}
\allowdisplaybreaks
\usepackage{txfonts}

\newcommand{\Vir}{\mathrm{Vir}}
\newcommand{\Z}{\mathbb{Z}}
\newcommand{\C}{\mathbb{C}}
\newcommand{\Lam}{|\Lambda\rangle}

\def\normOrd#1{\mathop{:}\nolimits\!#1\!\mathop{:}\nolimits}

\theoremstyle{plain}
\newtheorem{thm}{Theorem}[section]

\newtheorem{lem}[thm]{Lemma}
\newtheorem{prop}[thm]{Proposition}
\newtheorem{conj}[thm]{Conjecture}
\newtheorem{dfn}[thm]{Definition}
\newtheorem{exmp}[thm]{Example}
\newtheorem{re}[thm]{Remark}

\numberwithin{equation}{section}

\begin{document}
\title{Remarks on irregular conformal blocks and Painlev\'e III and II tau functions}
\author{Hajime Nagoya}
\address{Kanazawa University}
\email{nagoya@se.kanazawa-u.ac.jp}
\date{}
\maketitle

\begin{abstract}
We prove a conjecture 
on uniqueness and existence of the irregular vertex operators of rank $r$ 
introduced in the previous paper \cite{Nagoya ICB}. 
We also introduce ramified irregular vertex operators of the Virasoro algebra. As applications, we give 
conjectural formulas for series expansions at infinity 
of the tau functions of the second and third Painlev\'e equations
in terms of our ramified irregular conformal blocks. 
\end{abstract}

\section{Introduction}

This paper is a sequel to our paper \cite{Nagoya ICB}. There 
we considered  
two kinds of irregular vertex operators for 
the Virasoro algebra. One is an irregular vertex operator of rank zero from 
an irregular Verma module of rank $r$ to another irregular Verma module of rank $r$: 
\begin{equation*}
\Phi_{\Lambda',\Lambda}^{\Delta}(z): 
M_\Lambda^{[r]}\to M_{\Lambda'}^{[r]},
\end{equation*}
where $\Lambda=(\Lambda_{r},\ldots, \Lambda_{2r}), \,
\Lambda'=(\Lambda'_{r},\ldots, \Lambda'_{2r})\in \C^{r+1}$ and $\Delta\in \C$. 
We proved  that if $\Lambda_{2r}\neq 0$, then  
the commutation relations 
\begin{equation}\label{eq_comrel_intro}
\left[ L_n, \Phi_{\Lambda',\Lambda}^{\Delta}(z)\right]=
z^n\left( z \frac{\partial}{\partial z}+(n+1)\Delta\right) \Phi_{\Lambda',\Lambda}^{\Delta}(z)
\end{equation}
and the action on the irregular vector $|\Lambda\rangle$ 
\begin{equation}\label{eq_action_intro}
\Phi_{\Lambda',\Lambda}^{\Delta}(z)|\Lambda\rangle=
z^\alpha e^{\sum_{i=1}^r \beta_i/z^i} \sum_{m=0}^\infty v_m z^m, 
\end{equation}
where $v_0=|\Lambda'\rangle$, 
determine uniquely the irregular vertex operator $\Phi_{\Lambda',\Lambda}^{\Delta}(z)$ 
of rank zero. 

Another irregular vertex operators considered in the previous paper \cite{Nagoya ICB} 
is an irregular vertex operator of rank $r$ from a Verma module to an 
irregular Verma module of rank $r$: 
\begin{equation*}
\Phi^{[r],\lambda}_{\Delta,\Lambda}(z): M^{[0]}_{\Delta} \to M^{[r]}_{\Lambda}, 
\end{equation*}
where $\lambda=(\lambda_0,\ldots, \lambda_r)\in \C^{r+1}$. 
 For this case, the action of 
the irregular  vertex operator of rank $r$ 
on the highest weight vector is same to the one for an irregular vertex operator 
of rank zero, however, the commutation relations are much more complicated. 
See Definition \ref{def_IVO}. We conjectured in \cite{Nagoya ICB} 
that the defining 
commutation relations, in other words, OPE, and 
the action on the highest weight vector 
still determine uniquely the 
irregular  vertex operator of rank $r$.

Initially, $\Phi^{[r],\lambda}_{\Delta,\Lambda}(z)$ of rank $r$ 
was introduced in  
\cite{Nagoya Sun 2010} as a confluent version of the primary field 
\begin{equation}\label{eq_confluent_primary_field}
\Phi^{[r]}_\lambda(z)=:\exp\left(\sum_{i=0}^r \frac{\lambda_i}{i!}
\frac{\partial^i\varphi(z)}{\partial z^i}\right):.  
\end{equation}
Here, $\varphi(z)$ is the free field. In this case, $\alpha$, $\beta_1$, \ldots,
 $\beta_r$  and $\Lambda$ 
are determined by $\lambda$ and $\Delta$. Applying the screening operator 
to the confluent primary field \eqref{eq_confluent_primary_field}, we obtain 
an infinite number of irregular vertex operators of rank $r$ parametrized 
by positive integer parameter $p$. One of our aim is to generalize confluent primary fields 
so that we can replace the integer parameter 
$p$ with a complex parameter. 

In this paper, we prove the conjecture on uniqueness and existence 
of the irregular  vertex operators of rank $r$. Consequently, 
 we have established a construction of 
formal power series expansions 
of irregular conformal blocks at the irregular singular point $z$, 
because  formal power series expansions 
of irregular conformal blocks are given as expectation values of irregular 
vertex operators.  We note that the previous results in \cite{Nagoya ICB} on 
irregular vertex operators of rank zero  yield formal power series expansions 
of irregular conformal blocks at zero and infinity. 

Following a remarkable discovery by \cite{GIL12} that the tau function of the sixth
Painlev\'e equation admits 
a Fourier expansion in terms of four point conformal blocks in the two dimensional 
conformal field theory, 
in the previous paper \cite{Nagoya ICB}, 
we also conjectured that the tau functions of 
the fourth and fifth Painlev\'e equations are expressed as Fourier transforms of 
irregular conformal blocks with the central charge $c=1$, which are defined 
by using irregular vertex operators  $\Phi_{\Lambda',\Lambda}^{\Delta}(z)$ of rank zero. 
As seen in \cite{GIL13}, \cite{BLMST}, we expect that series expansions  
at the fixed singular points 
of the tau functions 
of Painlev\'e equations  
 are expressed by  
Fourier transforms of 
irregular conformal blocks with the central charge $c=1$. 
However, we had not succeeded  to 
define irregular conformal blocks of the Virasoro algebra 
for series expansions  at infinity of the tau functions of 
the first, second and third Painlev\'e equations.

It is known that an irregular Verma module $M_\Lambda^{[r]}$ is irreducible 
if and only if $\Lambda_{2r}\neq 0$ or $\Lambda_{2r-1}\neq 0$ \cite{FJK}, \cite{LGZ}. 
We are concerned here with irregular vertex operators from an  irregular 
Verma module  $M_\Lambda^{[r]}$ to $M_{\Lambda'}^{[r]}$, where   
$\Lambda_{2r}=\Lambda^{'}_{2r}=0$, $\Lambda_{2r-1}\neq 0$, $\Lambda^{'}_{2r-1}\neq 0$. 
Based on expectations that such cases were related to {\it ramified } irregular singularities, 
we formulate conditions that  irregular vertex operators exist uniquely in that case. 
We see  that these operators yield irregular conformal blocks expanded 
by $z^{1/2}$, while the indices of the generators $L_n$ of the Virasoro algebra remain integers. 
In particular, our examples include asymptotic expansions of the Bessel function and the Airy 
function at infinity. Furthermore, we conjecture that 
Fourier transforms of our newly introduced ramified irregular conformal blocks 
give series expansions at infinity of the tau functions of the second and third Painlev\'e equations. 
By \cite{BLMST},  ramified irregular conformal blocks with the 
central charge $c=1$ are expected to be equal to magnetic and dyonic 
Nekrasov partition functions for Argyres-Douglas theories at self-dual Omega 
background $\epsilon_1+\epsilon_2=0$. It may be interesting to 
explore what are counterparts of ramified irregular conformal blocks 
in Argyres-Douglas theories. A relation between $H_1$ Argyres-Douglas theory and 
Painlev\'e II was investigated in \cite{GG}. 

The plan of the paper is as follows. 
In Section 2, we prove the conjecture on uniqueness and existence 
of the irregular  vertex operators of rank $r$. 
In Section 3, we give a definition of  
ramified irregular vertex operators and conjectures on Fourier transforms of the tau functions  
of  the second and third Painlev\'e equations in terms of ramified irregular conformal blocks. 
%%%%%%%%%%%%%%%%%%%%%%%%%%%%

\section{Irregular vertex operators of rank $r$}
Let us denote the Virasoro algebra by $\mathrm{Vir}=\bigoplus_{n\in\mathbb{Z}}\mathbb{C}L_n
\oplus \mathbb{C}c$ with the commutation 
relations 
\begin{align*}
[L_m,L_n]=&(m-n)L_{m+n}+\delta_{m+n,0}\frac{n^3-n}{12}c, 
\\
[L_m, c]=&0.  
\end{align*}
For a non-negative integer $r$, an irregular Verma module $M_\Lambda^{[r]}$ is a representation of $\Vir$ 
with the vector $\Lam$ ($\Lambda=(\Lambda_r,\Lambda_{r+1},\ldots, \Lambda_{2r})\in \C^{r+1}$) 
such that, 
\begin{align*}
L_n\Lam= \Lambda_n\Lam \quad (n=r, r+1,\ldots, 2r), 
\end{align*}
and $M_\Lambda^{[r]}$ is spanned by linearly independent vectors of the form
\begin{align*}
L_{-i_1+r}\cdots L_{-i_k+r}\Lam \quad (i_1\ge \cdots \ge i_k>0). 
\end{align*}
We call  $\Lam$ the irregular vector. For $r=0$, $M_\Lambda^{[0]}$ is a Verma module with 
the highest weight vector $\Lam$ ($\Lambda\in\C$), which we denote by $M_\Delta$ ($\Delta\in \C$) usually.  Put $L_{-\mu}=L_{-\mu_1+r}\cdots L_{-\mu_k+r}$ 
for a partition $\mu=(\mu_1,\ldots,\mu_k)$.  We set $|\mu|=\mu_1+\cdots+\mu_k$ 
and denote by $\mathbb{Y}$ the set of all partitions. 

Let us recall the definition of an irregular vertex operator of rank $r$ introduced in \cite{Nagoya ICB}. 
Put $c=1-12\rho^2$. 
\begin{dfn}\label{def_IVO}
Let      $\alpha\in\C$, $\beta=(\beta_1,\ldots,\beta_r)\in \C^r$ and 
$\Lambda=(\Lambda_r,\ldots, \Lambda_{2r})$ in $\C^{r+1}$.  
  An irregular vertex operator  $\Phi^{[r],\lambda}_{\Delta,\Lambda}(z)$ of rank $r$  
  with variables $z$, $\lambda=(\lambda_0,\ldots, \lambda_r)$ 
 is 
a linear operator from $M^{[0]}_{\Delta}$ to $ M^{[r]}_{\Lambda}$ satisfying 
\begin{align}
\left[L_n,\Phi^{[r],\lambda}_{\Delta,\Lambda}(z)\right]=&z^{n+1}\partial_z\Phi^{[r],\lambda}_{\Delta,\Lambda}(z)
+\sum_{i=0}^{r-1}
\begin{pmatrix}
n+1
\\
i+1
\end{pmatrix}
z^{n-i}D_{i}\Phi^{[r],\lambda}_{\Delta,\Lambda}(z)\nonumber 
\\
&+\frac{1}{ 2}\sum_{i,j=0}^r \lambda_i\lambda_j\begin{pmatrix}n+1\\i+j+1\end{pmatrix}
z^{n-i-j}\Phi^{[r],\lambda}_{\Delta,\Lambda}(z)
-\rho\sum_{i=0}^r\begin{pmatrix}
n+1
\\
i+1
\end{pmatrix}(i+1)\lambda_iz^{n-i}\Phi^{[r],\lambda}_{\Delta,\Lambda}(z), \label{comrel_rankr}
\end{align}
where $D_k=\sum_{p=1}^{r-k}p\lambda_{p+k}\frac{\partial}{\partial \lambda_p}$
and 
\begin{equation}\label{eq_VO_act_0r}
\Phi^{[r],\lambda}_{\Delta,\Lambda}(z)|\Delta\rangle=z^\alpha \exp\left(\sum_{n=0}^r 
\frac{\beta_n}{z^n}\right)\sum_{m=0}^\infty v_mz^m,
\end{equation}
where $v_0=|\Lambda\rangle$, $v_m\in M^{[r]}_{\Lambda}$ ($m\ge 1$). \qed 
  \end{dfn}
Notice that the commutation relations 
\eqref{comrel_rankr}
of the rank $r$ vertex operator  and the condition \eqref{eq_VO_act_0r} 
imply  
\begin{align}
\label{L_n action part 1}
L_n v_m=&\delta_{n,0}\Delta v_m
+(\alpha+m-n-(n+1)\rho\lambda_0+(n+1)D_0)v_{m-n}
\\
&+\sum_{i=1}^{r-1}\begin{pmatrix}
n+1
\\
i+1
\end{pmatrix}D_{i}v_{m-n+i}-
\sum_{i=1}^r\left(i\beta_i+(i+1)\rho\lambda_i\begin{pmatrix}
n+1
\\
i+1
\end{pmatrix}\right)v_{m-n+i}
\nonumber
\\
&+\sum_{i,j=1}^r
\begin{pmatrix}
n+1
\\
i
\end{pmatrix}
D_{i-1}(\beta_j)v_{m-n-1+i+j}
+\frac{1}{2}\sum_{i,j=0}^r \lambda_i\lambda_j\begin{pmatrix}n+1\\i+j+1\end{pmatrix}v_{m-n+i+j},  \nonumber
\end{align}
for any $n\ge 0$, 
where $v_{-n}=0$ for $n>0$.

\subsection{Uniqueness}

In this section, we show that an irregular vertex operator of rank $r$ is unique in the following sense. 
\begin{thm} \label{thm_uniqueness}
Suppose that there exist complex parameters 
$\alpha$, $\beta_1$,\ldots, $\beta_r$, 
$\Lambda_r,\ldots, \Lambda_{2r}$  and 
elements $v_m$ of 
$M_\Lambda^{[r]}$ satisfying \eqref{L_n action part 1}. 
Then, $\alpha$, $\beta_1$, \ldots, $\beta_{r-1}$, $\Lambda_r$, \ldots, $\Lambda_{2r}$, 
$v_m$  
are uniquely determined by $\lambda$, $\beta_r$ and $\Delta$. Moreover, $v_m$ are expressed 
as $v_m=\sum_{|\mu|\le m(r+1)} c_\mu^{(m)}L_{-\mu}|\Lambda\rangle$ 
and the coefficients $c_\mu^{(m)}$ are polynomials in $\lambda_0,\ldots,\lambda_r,\lambda_r^{-1}$, 
$\beta_r$ and $\Delta$.  \qed
\end{thm}

Let us recall some properties of irregular vertex operators shown in \cite{Nagoya ICB}. 
\begin{prop}
For $i=0,\ldots,r-1$ and $k=1,\ldots,r$, we have 
\begin{equation}\label{D_i b_k 1r}
D_i(\beta_k)=(-1)^i(k+i)\beta_{k+i}, 
\end{equation}
where $\beta_k=0$ if $k>r$. \qed 
\end{prop}
From this proposition, $\beta_i$ ($i=1,\ldots,r$) should 
be expressed as functions of $\lambda_j$ ($j=i,\ldots, r$). 
\begin{prop}\label{prop_Lambda}
The parameters $\Lambda_n$ ($n=r,\ldots, 2r$) are solved as 
\begin{equation*}
\Lambda_n=\frac{1}{2}\sum_{i=0}^r \lambda_i\lambda_{n-i}
+\delta_{n,r}\left((-1)^{r+1}r\beta_r-(r+1)\rho\lambda_r\right).  
\end{equation*}
\qed
\end{prop}

The defining relations \eqref{L_n action part 1} for $v_m$  are for $n\ge 0$, however, 
it turns out that the actions of $L_n$ for $n\ge r$ on $v_m$ 
determine the irregular vertex operator 
uniquely. We use the following modified recursive relations for $L_n v_m$. Put
 $\widetilde{L}_n=L_n-\Lambda_n$, where $\Lambda_n=0$ if $n>2r$.  
\begin{prop}\label{prop-Lv}
For $n\ge r$, we have 
\begin{align}
\widetilde{L}_nv_m=&\sum_{i=r}^{2r}\begin{pmatrix}
n+1
\\
i+1
\end{pmatrix}\Lambda_{i}v_{m-n+i}
-\Lambda_n v_m
\label{eq_L_nv_m_n>=r}
+\sum_{j=0}^{r+1}\begin{pmatrix}
n+1\\j
\end{pmatrix}(-1)^{j-1}r\beta_rv_{m-n+r}
+\sum_{j=0}^{r}\begin{pmatrix}
n+1
\\
j
\end{pmatrix}(-1)^{j-1}
\sum_{i=1}^{r-1}i\beta_iv_{m-n+i}
\\
&+\sum_{i=0}^{r-1}\sum_{j=0}^i
\begin{pmatrix}
n+1
\\
i+1
\end{pmatrix}
\begin{pmatrix}
i+1\\j
\end{pmatrix}(-1)^jL_{i-j}v_{m-n+i-j}
-(n\alpha+(n+1)\Delta+n(m-n))v_{m-n}
\nonumber
\\
&+\sum_{i=1}^{r-1}\begin{pmatrix}
n+1\\i+1
\end{pmatrix}
(-1)^{i+1}(\alpha+m-n+(i+1)\Delta)v_{m-n}.  
\nonumber 
\end{align}
\qed
\end{prop}
\begin{proof}
From the defining relations \eqref{L_n action part 1} for $n=0,1,\ldots, r-1$, we have  
\begin{align*}
\left( D_k(\tilde{v}_m)+\tilde{v}_mL_k\right)
\Lam 
=&\sum_{i=0}^k
\begin{pmatrix}
k+1\\i
\end{pmatrix}
(-1)^iL_{k-i}v_{m-i}
+(-1)^k\sum_{i=1}^{k-1}i\beta_i v_{m-k+i}
\\
&+(-1)^{k+1}(\alpha+m-k+(k+1)\Delta-\delta_{k,0}(\alpha+\Delta))v_{m-k}
\end{align*}
for $k=0,1,\ldots, r-1$. 
Due to the relations above, we can eliminate the differential operators $D_n$ ($n=0,1,\ldots, r-1$) 
from the defining relations \eqref{L_n action part 1} for $n\ge r$. As a result, we obtain 
the recursive relations \eqref{eq_L_nv_m_n>=r}. 
\end{proof}

Proposition \ref{prop-Lv} implies that  
the action of $\widetilde{L}_n$ ($n=r,\ldots, 2r$) 
on $v_m$ is 
 a linear combination of $v_0$, $v_1$, \ldots, $v_{m-1}$ and  
the action of  $\widetilde{L}_{n+2r}$ ($n\ge 1$) 
on $v_m$ is a sum of $v_0$, $v_1$, \ldots, $v_{m-n}$. 
Put $\widetilde{L}_\nu=\widetilde{L}_{\nu_1+r}\cdots
\widetilde{L}_{\nu_k+r}$ for a partition $\nu=(\nu_1,\ldots, \nu_k)$ 
($\nu_i\ge \nu_{i+1}$).  
\begin{lem}[see Lemma 2.22 in \cite{Nagoya ICB}]\label{lem-LL}
For partitions $\nu$ and $\mu$ 
such that $ |\nu|\ge |\mu|$, we have 
\begin{equation*}
\widetilde{L}_\nu
L_\mu|\Lambda\rangle=\left\{\begin{matrix}
0 & (\nu\neq \mu),
\\
\left(2\Lambda_{2r}\right)^{\ell(\nu)}
\prod_{i=1}^{\ell(\nu)}\nu_i  & 
(\nu=\mu).
\end{matrix}\right.
\end{equation*}
\qed
\end{lem}

\begin{prop}\label{prop_v_expression}
We have 
\begin{equation}\label{eq_v_m_expression}
v_m=\sum_{|\nu|\le m(r+1)} c_\nu^{(m)}L_{-\nu}
\Lam. 
\end{equation}
\qed
\end{prop}
\begin{proof}
By the recursive relations \eqref{eq_L_nv_m_n>=r}, we have 
$\widetilde{L}_\nu v_m=0$ for $|\nu|>m(r+1)$. Hence, 
Lemma \ref{lem-LL} implies $v_m$ is expressed as \eqref{eq_v_m_expression}. 
\end{proof}
It is easy to see that 
by applying $\widetilde{L}_\nu$ ($1\le |\nu|\le m(r+1)$) to $v_m$, 
$c_\nu^{(m)}$ are expressed as polynomials 
in $\alpha$, $\beta_1,\ldots,\beta_r$, $\lambda_0,\ldots, \lambda_r,\lambda_r^{-1}$, 
$\Delta$ and $c_\emptyset^{(i)}$ ($1\le i\le m-1$).  We show below that 
$\alpha$, $\beta_1,\ldots, \beta_{r-1}$, $c_\emptyset^{(m)}$ are 
expressed as polynomials 
in $\beta_r$, $\lambda_0,\ldots, \lambda_r,\lambda_r^{-1}$ and  
$\Delta$. 

\begin{prop}\label{prop_v1}
We have 
\begin{equation*}
v_1=L_{-(1+r)+r}\Lam+\frac{(-1)^{r+2}}{2\Lambda_{2r}}r\beta_rL_{-1+r}\Lam+c_\emptyset^{(1)}\Lam. 
\end{equation*}
For $r=1$, 
\begin{equation*}
\alpha=-2\Delta+\frac{\beta_1\Lambda_1}{2\Lambda_2} 
\end{equation*}
and for $r>1$, 
\begin{equation*}
\beta_{r-1}=-\frac{r\beta_r\lambda_{r-1}}{(r-1)\lambda_r}. 
\end{equation*}
\qed 
\end{prop}
\begin{proof}
We apply $\widetilde{L}_\nu$ ($|\nu| \le r+1$) to $v_1$ in 
order from largest to smallest. Then, it is easy to see that 
$c_\nu^{(1)}$ ($|\nu|\le r+1$) and 
$\alpha$ in the case of $r=1$, and $\beta_{r-1}$ in the case of $r>1$ are uniquely determined  
due to \eqref{eq_L_nv_m_n>=r}. 
\end{proof}

Because the action of any $\widetilde{L}_\nu$ on $v_1$ kills the 
constant term $c_\emptyset^{(1)}\Lam$, we can not determine $c_\emptyset^{(1)}$ 
as a polynomial of the parameters from the computation of $v_1$.  
Later, it turns out that $c_\emptyset^{(1)}$ is determined when we compute $v_{r+1}$.

Let $X_m$ ($m\in \Z_{\ge 0}$) be defined recursively by 
\begin{equation*}
X_m=v_m-\sum_{i=\mathrm{max}(1,m-r)}^{m}c_\emptyset^{(i)} X_{m-i}, 
\end{equation*}
where $X_0=v_0=|\Lambda\rangle$. Put $X_m=0$ for $m\in\Z_{<0}$. 
Denote $X_m=\sum_{\nu}x_\nu^{(m)}L_{-\nu}|\Lambda\rangle$. By definition, the constant term of $X_m$ for $m\ge 1$ is equal to zero, 
namely, $x_\emptyset^{(m)}=0$ ($m\ge 1$). 

For $m\ge 1$ and  $n\ge r$, let   
$a_{m,n,\ell}$ be given by  the relation \eqref{eq_L_nv_m_n>=r} $\widetilde{L}_nv_m=\sum_{\ell=1}^{n} a_{m,n,\ell}v_{m-\ell}$.  
By definition, $a_{m, n,\ell}$ for $\ell\le n-1$ does not depend on $m$. 
In this case, we omit the index $m$ from $a_{m,n,\ell}$. 
\begin{lem}
(i) For $1\le m\le r$ and $n> r$, we have 
\begin{align}
\widetilde{L}_nX_m=&\sum_{\ell=1}^{m} a_{n,\ell} X_{m-l}. \label{eq_LX1}
\end{align}

(ii) For $1\le m<r$,  
we have 
\begin{align}\label{eq_LX2}
\widetilde{L}_rX_m=&\sum_{\ell=1}^m a_{r,\ell} X_{m-l}. 
\end{align}

(iii) We have 
\begin{align}\label{eq_LX3}
\widetilde{L}_rX_r=&\sum_{\ell=1}^{r-1} a_{r,\ell} X_{r-l}
+a_{r,r,r}v_0. 
\end{align}

(iv) For $m\ge r+1$ and $n>r$, we have 
\begin{align}\label{eq_LX4}
\widetilde{L}_n X_m=\sum_{\ell=1}^r a_{n,\ell}X_{m-\ell}
+f,
\end{align}
where $f$ is a some polynomial of $\alpha$, $\beta_1,\ldots,\beta_r$, $\lambda_0,\ldots, \lambda_r,\lambda_r^{-1}$, 
$\Delta$ and $c_\emptyset^{(i)}$ ($1\le i\le m-r-1$). 

(v) For $m\ge r+1$, we have 
\begin{align}\label{eq_LX5}
\widetilde{L}_r X_m=\sum_{\ell=1}^{r-1} a_{r,\ell}X_{m-\ell}
+a_{m,r,r}v_{m-r}-c_\emptyset^{(m-r)}a_{r,r,r}v_0
+f,
\end{align}
where $f$ is a some polynomial of $\alpha$, $\beta_1,\ldots,\beta_r$, $\lambda_0,\ldots, \lambda_r,\lambda_r^{-1}$, 
$\Delta$ and $c_\emptyset^{(i)}$ ($1\le i\le m-r-1$). 
\qed 
\end{lem}
\begin{proof}
(i) By induction. Since $X_1=v_1-c_\emptyset^{(1)}v_0$, $\widetilde{L}_nX_1=a_{n,1}v_0=a_{n,1}X_0$ 
for $n>r$. 
Suppose the identity \eqref{eq_LX1} is true for $m=1,\ldots, k-1$. Then, for $n>r$ we have 
\begin{align*}
\widetilde{L}_nX_k=&\sum_{\ell=1}^k a_{n,\ell}v_{k-\ell}
-\sum_{i=1}^mc_\emptyset^{(i)}\sum_{\ell=1}^m a_{n,\ell}X_{m-i-\ell}
\\
=&\sum_{\ell=1}^k a_{n,\ell}v_{k-\ell}
-\sum_{\ell=1}^{k-1} a_{n,\ell}\sum_{i=1}^{k-\ell}c_\emptyset^{(i)}X_{k-\ell-i}
\\
=&\sum_{\ell=1}^k a_{n,\ell}X_{k-\ell}. 
\end{align*}
We can show (ii), (iii), (iv) and (v) in the same way. 
\end{proof}

{\bf Proof of Theorem \ref{thm_uniqueness}.}
We show that $\beta_{r-m}$
 ($m=1,\ldots,r-1$), $\alpha$, $c_\emptyset^{(m-r)}$ are determined using $\widetilde{L}_rX_{m}$  
 ($m=1,2,\ldots$), respectively. We have already shown that $\beta_{r-1}$ 
 is a polynomial of $\beta_r$, $\lambda_r$, $\lambda_r^{-1}$ and 
 $x_\nu^{(1)}$ ($\nu\in \mathbb{Y}$) are polynomials of $\beta_r$ and 
  $\lambda_r^{-1}$. 
 
Suppose that for $m\le r-1$, $x_\nu^{(i)}$ ($i\le m-1$) and $\beta_{r-1},\ldots, \beta_{r-m+1}$
 are expressed as polynomials of $\beta_r$, $\lambda_0,\ldots, \lambda_r,\lambda_r^{-1}$ 
 and $\Delta$. Then, by Proposition \ref{prop-Lv}, Lemma \ref{lem-LL}
 and the identity \eqref{eq_LX1}, $x_\nu^{(m)}$ 
 ($\nu\in\mathbb{Y}$) are expressed as 
polynomials of $\beta_r$, $\lambda_0,\ldots, \lambda_r,\lambda_r^{-1}$ 
 and $\Delta$. 
Since by definition $a_{r,\ell}$ for $\ell=1,\ldots, r-m-1$ are polynomials 
of  $\beta_{r-m+1},\ldots, \beta_r$, $\lambda_0,\ldots, \lambda_r,\lambda_r^{-1}$
and 
 \begin{equation*}
 a_{r,m}=\sum_{j=0}^r\begin{pmatrix}
r+1
\\
j
\end{pmatrix}(-1)^{j-1}(r-m)\beta_{r-m}, 
 \end{equation*}
 from the identity \eqref{eq_LX2} we can determine $\beta_{r-m}$ as 
 a polynomial of $\beta_r$, $\lambda_0,\ldots, \lambda_r,\lambda_r^{-1}$ 
 and $\Delta$. 
 
 Similarly, we obtain $x_\nu^{(r)}$ 
 ($\nu\in\mathbb{Y}$)  as 
polynomials of $\beta_r$, $\lambda_0,\ldots, \lambda_r,\lambda_r^{-1}$ 
 and $\Delta$. Since by definition we have 
 \begin{equation*}
 a_{r,r,r}=-(r\alpha+(r+1)\Delta)
+\sum_{i=1}^{r-1}\begin{pmatrix}
r+1\\i+1
\end{pmatrix}
(-1)^{i+1}(\alpha+(i+1)\Delta), 
 \end{equation*}
from the identity \eqref{eq_LX3} we obtain $\alpha$ as a polynomial of $\beta_r$, $\lambda_0,\ldots, \lambda_r,\lambda_r^{-1}$ 
 and $\Delta$. 
 
 Suppose that for $m\ge r+1$, $x_\nu^{(i)}$ for $i\le m-1$ and 
 $c_\emptyset^{(i)}$ for $i\le m-r-1$ 
 are expressed as polynomials of $\beta_r$, $\lambda_0,\ldots, \lambda_r,\lambda_r^{-1}$ 
 and $\Delta$. Then, by Proposition \ref{prop-Lv}, Lemma \ref{lem-LL}
 and the identity \eqref{eq_LX4}, $x_\nu^{(m)}$ 
 ($\nu\in\mathbb{Y}$) are expressed as 
polynomials of $\beta_r$, $\lambda_0,\ldots, \lambda_r,\lambda_r^{-1}$ 
 and $\Delta$. 
Since by definition we have 
\begin{align*}
a_{m,r,r}-a_{r,r,r}=&-r(m-r)+\sum_{i=1}^{r-1}\begin{pmatrix}
r+1\\i+1
\end{pmatrix}(-1)^{i+1}(m-r)
\\
=&(m-r)(-1)^r, 
\end{align*}
from the identity \eqref{eq_LX5} we obtain $c_\emptyset^{(m-r)}$ 
as a polynomial of $\beta_r$, $\lambda_0,\ldots, \lambda_r,\lambda_r^{-1}$ 
 and $\Delta$. 
\qed

\subsection{Existence}

From the facts shown in the previous subsection, we have $v_m$ ($m\ge 1$) in 
the irregular vertex operator of rank $r$ 
as polynomials of $\lambda_0,\ldots,\lambda_r,\lambda_r^{-1}$, 
$\beta_r$ and $\Delta$. In order to prove existence of 
the irregular vertex operator of rank $r$, we should prove the identity \eqref{L_n action part 1} 
for any $n\ge 0$ and $v_m$ ($m\ge 0$). We did not obtain a direct proof of it because we have to 
deal with differential equations involved in the identity \eqref{L_n action part 1}. 
However, by the free field realization we know existence of irregular vertex operators  
of rank $r$ when $\beta_r$ is parametrized by a positive integer $p$. Hence, 
the identity \eqref{L_n action part 1} for any $n\ge 0$ and $v_m$ ($m\ge 0$) holds 
for general $\beta_r$. 
Therefore, we arrive at 
\begin{thm} 
For given parameters  
$\beta_r$ and $\Delta$, the irregular vertex operator of rank $r$ uniquely exists. 
Namely,  $\alpha$, $\beta_1$, \ldots, $\beta_{r-1}$, $\Lambda_r$, \ldots, $\Lambda_{2r}$, 
$v_m$  
are uniquely determined by $\lambda$, $\beta_r$ and $\Delta$. Moreover, $v_m$ are expressed 
as $v_m=\sum_{|\mu|\le m(r+1)} c_\mu^{(m)}L_{-\mu}|\Lambda\rangle$ 
and the coefficients $c_\mu^{(m)}$ are polynomials in $\lambda_0,\ldots,\lambda_r,\lambda_r^{-1}$, 
$\beta_r$ and $\Delta$.  \qed 
\end{thm}

%%%%%%%%%%%%%%%%%%%%%%%%

\section{Ramified irregular vertex operators}

In this section, we introduce a ramified irregular vertex operator  
of the Virasoro algebra. Using a ramified irregular vertex operator, 
we can construct a ramified irregular conformal block. 
As applications, we give 
conjectural formulas for series expansions of the tau functions of the second 
and third Painlev\'e equations  
in terms of our ramified irregular conformal blocks.

\begin{dfn}
An irregular vertex operator of ramified type
$\Phi^{\Delta}_{\Lambda',\Lambda}(z): 
M^{[r]}_{\Lambda}\to M^{[r]}_{\Lambda'}$ is a 
linear operator satisfying 
\begin{align}
&[L_n, \Phi^{ \Delta}_{\Lambda,\Lambda'}(z)]=z^n \left(z\frac{\partial}{\partial z}+
(n+1)\Delta\right)\Phi^{ \Delta}_{\Lambda,\Lambda'}(z), \label{comrel_rank0}
\\
&\Phi^{\Delta}_{\Lambda,\Lambda'}(z)\Lam=z^\alpha \exp\left(\sum_{i=1}^{2r-1} 
\frac{\beta_i}{z^{i/2}}\right)
\sum_{m=0}^\infty v_mz^{m/2},\label{eq_VO_act_rr}
\end{align}
where $v_0=|\Lambda'\rangle$, $v_m\in M^{[r]}_{\Lambda'}$ ($m\ge 1$) and 
$\alpha$, $\beta_i$ ($i=1,\ldots, 2r-1$) are complex parameters. \qed
  \end{dfn}

An irregular Verma module $M_\Lambda^{[r]}$ is irreducible 
if and only if $\Lambda_{2r}\neq 0$ or $\Lambda_{2r-1}\neq 0$ \cite{FJK}, \cite{LGZ}. 
We remark that the irregular Verma module 
$M_{\Lambda'}^{[r]}$ is irreducible if $\Lambda'_{2r}\neq 0$, and 
the conditions \eqref{comrel_rank0} and \eqref{eq_VO_act_rr} 
implies the identity $\Lambda_{2r}'=\Lambda_{2r}$.   
  
\begin{conj}
Let $r$ be a positive integer. If $\Lambda_{2r-1}\neq 0$ and $\Lambda_{2r}=0$, then 
the irregular vertex operator of ramified type
$\Phi^{\Delta}_{\Lambda',\Lambda}(z)$ 
exists. Here  $\Lambda'=\Lambda$ and   
$v_m=\sum_{|\mu|\le m}c_\mu^{(m)}L_{-\mu}
\Lam$. For any $\mu$, $c_\mu^{(m)}$  
is expressed as a polynomial of $\alpha$, $\beta_1$,\ldots, 
$\beta_{2r-1}$, 
$\Lambda_r$,\ldots, $\Lambda_{2r-1}$, $\Lambda_{2r-1}^{-1}$  
and $c_\emptyset^{(k)}$ ($k\le m-1$). \qed
\end{conj}

Since the irregular vertex operator of ramified type 
$\Phi^{\Delta}_{\Lambda',\Lambda}(z)$ 
depends on the parameters $\alpha$, 
$\beta_1$,\ldots, 
$\beta_{2r-1}$, 
$\Lambda_r$,\ldots, $\Lambda_{2r-1}$,   
and complex numbers $c_\emptyset^{(k)}$ ($k\le m-1$), 
it should be denoted by $\Phi^{\Delta}_{\Lambda,\Lambda}(\alpha, \beta, c_\emptyset;z)$. 

We define actions of $L_{-n}$ for any positive integer $n$ 
on a linear operator $\Phi(z)$ as follows. 
\begin{align*}
L_{-1}\cdot \Phi(z)
=&\frac{\partial}{\partial z}
\Phi(z)
\\
L_{-n}\cdot \Phi(z)
=&\normOrd{\frac{1}{(n-2)!}\partial_z^{n-2}(T(z))
\Phi(z)}
\\
=&\frac{1}{(n-2)!}
\left(\partial_z^{n-2}(T_-(z))\Phi(z)
+\Phi(z)\partial_z^{n-2}(T_+(z))\right),
\end{align*}
where $n\ge 2$ and 
\begin{align*}
T_-(z)=\sum_{n\le -2}L_nz^{-n-2},\quad 
T_+(z)=\sum_{n\ge -1}L_nz^{-n-2}. 
\end{align*}
We define the descendants of the irregular 
vertex operator of ramified type 
$\Phi_{\Lambda,\Lambda}^\Delta(\alpha, \beta, c_\emptyset;v,z)$ 
for $v\in M_\Delta$ 
by 
\begin{align*}
\Phi_{\Lambda,\Lambda}^\Delta(\alpha, \beta, c_\emptyset;|\Delta\rangle,z)=&
\Phi^{\Delta}_{\Lambda,\Lambda}(\alpha, \beta, c_\emptyset;z), 
\\
\Phi_{\Lambda,\Lambda}^\Delta(\alpha, \beta, c_\emptyset;L_{-\lambda}|\Delta\rangle,z)
=&L_{-\lambda}\cdot 
\Phi^{\Delta}_{\Lambda,\Lambda}(\alpha, \beta, c_\emptyset;z).
\end{align*}

\begin{dfn}
An irregular vertex operator  of ramified type
$\Phi^{\Delta}_{\Lambda,\Lambda}(\alpha, \beta, c_\emptyset;z)$ 
is called singular if it satisfies 
\begin{equation}
\Phi^{\Delta_{p,q}}_{\Lambda,\Lambda}(\alpha, \beta, c_\emptyset;
\chi_{p,q}, z)=0
\end{equation}
for the singular vector $\chi_{p,q}$ of level $pq$ 
in $M_{\Delta_{p,q}}$.  \qed
\end{dfn}

It is known that for positive integers $p,q$, 
a singular vector $\chi_{p,q}$ of level $pq$ 
exists in $M_{\Delta_{p,q}}$, where 
\begin{align*}
&c=13-6\left(t+\frac{1}{t}\right),\quad 
\Delta_{p,q}=\frac{(pt-q)^2-(t-1)^2}{4t}. 
\end{align*}

\begin{conj}
For any positive integers $p,q$, 
a singular irregular vertex operator  of ramified type 
$\Phi^{\Delta_{p,q}}_{\Lambda,\Lambda}(\alpha, \beta, c_\emptyset;z)$ 
exists. Here  the parameters $\alpha$, $\beta_1$, \ldots, 
$\beta_{2r-1}$ and $c_\emptyset^{(m)}$ ($m\in\Z_{\ge 1}$) 
are expressed as polynomials of $c$, $\Lambda_r$, 
\ldots, $\Lambda_{2r-1}$, $\Lambda_{2r-1}^{-1}$. Moreover, 
the number of such sets $\{\alpha, \beta,c_\emptyset\}$ 
with multiplicity 
is $pq$. \qed
\end{conj}
We denote a singular irregular vertex operator of ramified type 
by $\Phi^{\Delta_{p,q},i}_{\Lambda,\Lambda}(z)$ 
($i=1,\ldots, pq$) and $\beta_{2r-1}$ of $\Phi^{\Delta_{p,q},i}_{\Lambda,\Lambda}(z)$ 
by  $\beta_{2r-1}^{p,q,i}$. 

\begin{conj}
An irregular vertex operator of ramified type 
$\Phi^{\Delta}_{\Lambda,\Lambda}(\alpha, \beta, c_\emptyset;z)$ uniquely exists. 
Here the coefficients $c_\lambda^{(m)}$ are polynomials of $c$, $\Delta$, $\beta_{2r-1}$, $\Lambda_r$, 
\ldots, $\Lambda_{2r-1}$, $\Lambda_{2r-1}^{-1}$ and 
it is equal to the singular irregular 
vertex operator of ramified type 
$\Phi^{\Delta_{p,q},i}_{\Lambda,\Lambda}(z)$ 
when $\beta_{2r-1}=\beta_{2r-1}^{p,q,i}$ 
and $\Delta=\Delta_{p,q}$. \qed
\end{conj}
We denote such an irregular vertex operator 
as $\Phi_{\Lambda,\Lambda}^{\Delta,\beta_{2r-1}}(z)$ 
and call it a ramified irregular vertex operator. 

\subsection{Examples}
In this subsection, we give examples of ramified irregular vertex operators 
 and ramified irregular conformal blocks. 
Let $\tilde{v}_m$ be given by
\begin{equation*}
\Phi_{\Lambda,\Lambda}^{\Delta,\beta_{2r-1}}(z)|\Lambda\rangle
=z^\alpha \exp\left(\sum_{i=1}^{2r-1} 
\frac{\beta_i}{z^{i/2}}\right)
\sum_{m=0}^\infty \tilde{v}_m|\Lambda\rangle  z^{m/2}. 
\end{equation*}

\begin{exmp}
A ramified irregular vertex operator of half rank for $\Lambda_1=1$, $\beta_1=\beta$
  is given as follows. 
\begin{align*}
\alpha=&\frac{\beta^2}{32}-\frac{3 \Delta }{2},
\\
\tilde{v}_1=&\frac{\beta ^3}{256}+\frac{\beta}{64}  (c-4 \Delta +1)
%\frac{\beta ^3}{256}-\frac{\beta  \Delta }{16}-\frac{\beta  \left(3 t^2-7 t+3\right)}{32 t}
-\frac{\beta}{2}L_0,
\\
\tilde{v}_2=&\frac{\beta ^6}{131072}+\frac{\beta ^4 (c-4 \Delta +6)}{16384}
+\frac{\beta ^2 \left(3 c^2-24c \Delta +74 c+48 \Delta ^2-168 \Delta +103\right)}{24576}
+\frac{\Delta}{64}  (\Delta-c -2)
%\frac{\beta ^6}{131072}-\frac{\beta ^4 \Delta }{4096}+\frac{\beta ^2 \Delta ^2}{512}+\frac{\Delta ^2}{64}-\frac{\beta ^4 \left(6 t^2-19 t+6\right)}{16384 t}-\frac{\beta ^2 \Delta  \left(-3 t^2+10 t-3\right)}{512 t}
%+\frac{3 \Delta  \left(2 t^2-5 t+2\right)}{64 t}
\\
%-\frac{\beta ^2 \left(-9 t^4+76 t^3-149 t^2+76 t-9\right)}{2048 t^2}
&+\left(-\frac{\beta ^4}{512}-\frac{\beta^2}{128} (c-4 \Delta +13)+\frac{\Delta }{2}\right)L_0
+\frac{\beta ^2}{8}L_0^2, 
\\
\tilde{v}_3=& \frac{\beta ^9}{100663296}+\frac{\beta ^7 (c-4 \Delta +11)}{8388608}
+\frac{\beta ^5 \left(3 c^2-24 c \Delta +104 c+48 \Delta ^2-288 \Delta +397\right)}{6291456}
\\
&+\frac{\beta ^3 \left(3 c^3-36 c^2 \Delta +213 c^2+144 c \Delta ^2-1608 c \Delta +3793 c-192 \Delta ^3+2160 \Delta ^2-6660 \Delta +5951\right)}{4718592}
\\
&-\frac{\beta  \left(6 c^2 \Delta +7 c^2-30 c \Delta ^2+178 c \Delta -6 c+24 \Delta ^3-158 \Delta ^2+340 \Delta -37\right)}{24576}
\\
&+\left( -\frac{\beta ^7}{262144}
-\frac{\beta ^5 (c-4 \Delta +18)}{32768}
-\frac{\beta ^3 \left(3 c^2-24 c \Delta +146 c+48 \Delta ^2-552 \Delta +1359\right)}{49152}\right.
\\
&+\left.\frac{\beta}{384} \left(6 c \Delta +3 c-15 \Delta ^2+93 \Delta -13\right)\right)L_0-\frac{\beta}{6}L_{-1}-\frac{\beta^3}{48}L_0^3
\\
&+\left(\frac{\beta ^5}{2048}
+\frac{\beta^3}{512} (c-4 \Delta +25)
-\frac{\beta}{24}   (6 \Delta -1)\right)L_0^2. 
\end{align*}\qed
\end{exmp}

\begin{exmp}\label{ex_001/2}
A ramified irregular conformal block with two regular singular points and 
one ramified irregular singular point of half rank is given as follows. 
\begin{align*}
&\left\langle \Delta'\left|\  \Phi_{(1,0),(1,0)}^{\Delta,\beta}(z)\ \right|(1,0)\right\rangle
\\
&=z^{\beta^2/32-3\Delta/2} e^{\beta/\sqrt{z}}
\\
&\times \left\{\left(\frac{\beta ^3}{256}+\frac{\beta}{64} (-32 \Delta'+c-4 \Delta +1)\right)
z^{1/2}
\right. 
\\
&+\left(\frac{\beta ^6}{131072}+\frac{\beta ^4 (-32 \Delta'+c-4 \Delta +6)}{16384}\right.
\\
&\left.+\frac{\beta ^2 \left(-192 c \Delta'+768 \Delta  \Delta'+3072 \Delta'^2-2496 \Delta'+3 c^2-24 c \Delta +74 c+48 \Delta ^2-168 \Delta +103\right)}{24576}\right.
\\
&\left.+\frac{1}{64} \Delta  (32 \Delta'-c+\Delta -2)\right)z
\\
&+\left(\frac{\beta ^9}{100663296}+\frac{\beta ^7 (-32 \Delta'+c-4 \Delta +11)}{8388608}\right.
\\
&+\frac{\beta ^5 \left(-192 c \Delta'+768 \Delta  \Delta'+3072 \Delta'^2-3456 \Delta'+3 c^2-24 c \Delta +104 c+48 \Delta ^2-288 \Delta +397\right)}{6291456}
\\
&+\frac{\beta^3}{4718592}  \left(-288 c^2 \Delta'+2304 c \Delta  \Delta'+9216 c \Delta'^2-14016 c \Delta'-4608 \Delta ^2 \Delta'-36864 \Delta  \Delta'^2\right.
\\
&+52992 \Delta  \Delta'-98304 \Delta'^3+230400 \Delta'^2-130464 \Delta'+3 c^3-36 c^2 \Delta +213 c^2+144 c \Delta ^2
\\
&\left.-1608 c \Delta +3793 c-192 \Delta ^3+2160 \Delta ^2-6660 \Delta +5951\right)
\\
&-\frac{\beta}{24576}  \left(-384 c \Delta  \Delta'-192 c \Delta'+960 \Delta ^2 \Delta'+6144 \Delta  \Delta'^2-5952 \Delta  \Delta'-1024 \Delta'^2\right.
\\
&\left.\left.\left.+832 \Delta'+6 c^2 \Delta +7 c^2-30 c \Delta ^2+178 c \Delta -6 c+24 \Delta ^3-158 \Delta ^2+340 \Delta -37\right)\right)z^{3/2}+O(z^2)\right\}. 
\end{align*}
\qed 
\end{exmp}

\begin{exmp}
A ramified irregular vertex operator of  rank $3/2$ for
$\beta_1=\beta_2=\Lambda_2=0$ and $\Lambda_3=1$ 
 is given as follows. 
\begin{align*}
\alpha=&\frac{27 \beta ^2}{32}-\frac{5 \Delta }{2}, 
\\
\tilde{v}_1=&-\frac{3}{2} \beta  L_1, 
\\
\tilde{v}_2=&\frac{9}{8} \beta ^2 L_1^2,
\\
\tilde{v}_3=&\frac{153 \beta ^3}{256}-\frac{\beta}{192}  (5 c+108 \Delta -11)
-\frac{9}{16} \beta ^3 L_1^3-\frac{\beta}{2}   L_0, 
\\
\tilde{v}_4=&\left(-\frac{459 \beta ^4}{512}
+\frac{1}{128} \beta ^2 (5 c+108 \Delta -95)+\frac{\Delta }{2}\right)L_1
+\frac{27}{128} \beta ^4 L_1^4+\frac{3}{4} \beta ^2 L_0L_1,
\\
\tilde{v}_5=&-\frac{3}{10} \beta  L_{-1}+\left(\frac{1377 \beta ^5}{2048}
-\frac{3}{512} \beta ^3 (5 c+108 \Delta -179)
-\frac{3}{40} \beta  (10 \Delta -1)\right)L_1^2
\\
&-\frac{81 \beta ^5 }{1280}L_1^5-\frac{9}{16} \beta ^3 L_0L_1^2,
\\
\tilde{v}_6=&\frac{23409 \beta ^6}{131072}-\frac{3 \beta ^4 (85 c+1836 \Delta -3562)}{16384}
\\
&
+\frac{\beta ^2 \left(25 c^2+1080 c \Delta -6050 c+11664 \Delta ^2-58104 \Delta +15781\right)}{73728}+\frac{1}{192} \Delta  (5 c+11 \Delta -22)
\\
&-\left(\frac{153 \beta ^4}{512}-\frac{\beta ^2 (25 c+540 \Delta -1387)}{1920}-\frac{\Delta }{2}\right)L_0+\frac{9}{20} \beta ^2 L_{-1}L_1+\frac{1}{8} \beta ^2 L_0^2
\\
&-\left(\frac{1377 \beta ^6}{4096}-\frac{3 \beta ^4 (5 c+108 \Delta -263)}{1024}
-\frac{9}{80} \beta ^2 (5 \Delta -1)\right)L_1^3
+\frac{81 \beta ^6 }{5120}L_1^6+\frac{9}{32} \beta ^4 L_0L_1^3. 
\end{align*}
\qed
\end{exmp}

\begin{exmp}\label{ex_03/2}
A ramified irregular conformal blocks 
with two one regular singular point and one ramified irregular singular point of 
rank $3/2$ for 
$\beta_1=\beta_2=\Lambda_2=0$ and $\Lambda_3=1$ 
 is given as follows. 
 \begin{align*}
 &\left\langle 0\left|\ \Phi_{(0,1,0),(0,1,0)}^{\Delta,\beta}(z)\ \right|(0,1,0)\right\rangle
 \\
 &=z^{27\beta^2/32-5\Delta/2}e^{\beta/z^{3/2}}
 \\
 &\times \left\{\left(
\frac{153 \beta ^3}{256}-\frac{1}{192} \beta  (5 c+108 \Delta -11) \right)z^{3/2}
 \right.
 \\
 &+\left(\frac{23409 \beta ^6}{131072}-\frac{3 \beta ^4 (85 c+1836 \Delta -3562)}{16384}\right.
 \\
 &\left.+\frac{\beta ^2 \left(25 c^2+1080 c \Delta -6050 c+11664 \Delta ^2-58104 \Delta +15781\right)}{73728}+\frac{1}{192} \Delta  (5 c+11 \Delta -22)\right)z^3
 \\
 &+o(z^3). 
 \end{align*}
\qed
\end{exmp}

\subsection{Conjectures}

It is well known that the Painlev\'e equations are derived from monodromy preserving 
deformation for 2 by 2 linear systems  \cite{JM}. The corresponding linear systems 
$\mathrm{L_J}$ ($\mathrm{J=I,II,III,IV,V,VI}$)
admit the following degeneration scheme:
\begin{equation*}
\begin{diagram}
\node{\mathrm{L_{VI}}\atop (0,0,0,0)}\arrow{e}
\node{\mathrm{L_{V}}\atop (0,0,1)}\arrow{e}\arrow{s}
\node{\mathrm{L_{III}}\atop (0,0,1/2)}\arrow{s}%\arrow{e}
%\node{\mathrm{P_{III}^{D_7}}}\arrow{e}\arrow{se}
%\node{\mathrm{P_{III}^{D_8}}}
\\
\node[2]{\mathrm{L_{IV}}\atop (0,2)}\arrow{e}
\node{\mathrm{L_{II}}\atop (0,3/2)}\arrow{e}
\node{\mathrm{L_I}\atop (5/2)}
\end{diagram}
\end{equation*}
as explained in \cite{OO} in the case of single equations of the second order. 
Here, the numbers express the Poincar\'e ranks of singular points. For example, 
the tuple $(0,0,0,0)$ of $\mathrm{L_{VI}}$ means that $\mathrm{L_{VI}}$ is a $2$ by $2$ 
linear system with four regular singular points. 
We note that monodromy preserving deformation
 of $2$ by $2$ linear systems of type $(1,1)$, $(3)$ also yield $\mathrm{P_{III}}$, 
 $\mathrm{P_{II}}$, respectively.

As mentioned in Introduction,  a Fourier expansion of 
 the tau function of $\mathrm{P_{VI}}$ is expressed in terms of 
a four-point 
regular conformal block, namely, a conformal block of type $(0,0,0,0)$ 
\cite{GIL12}. In \cite{Nagoya ICB}, using 
  irregular conformal blocks of type $(0,0,1)$, $(0,2)$, 
we presented  Fourier expansions of 
 the tau functions of $\mathrm{P_V}$, $\mathrm{P_{IV}}$, respectively. 
 By noticing  correspondence between the types of singularities of the linear systems 
 $\mathrm{L_{VI}}$, $\mathrm{L_V}$, $\mathrm{L_{IV}}$, and 
 conformal blocks for the tau functions of 
 $\mathrm{P_{VI}}$, 
 $\mathrm{P_V}$, $\mathrm{P_{IV}}$, it is natural to expect that using 
 ramified irregular conformal blocks, we may give  
 Fourier expansions of 
 the tau functions of $\mathrm{P_{III}}$ and $\mathrm{P_{II}}$. 

The third and second Painlev\'e equations are the following second order 
ordinary differential equations:
\begin{align*}
&\mathrm{P_{III}} : \quad \frac{d^2 \lambda}{dt^2}=
\frac{1}{\lambda}\left(\frac{d\lambda}{dt}\right)^2-\frac{1}{t}\frac{d\lambda}{dt}
+\frac{\lambda^2}{4t^2}(\gamma \lambda+\alpha)+\frac{\beta}{4t}
+\frac{\delta}{4\lambda},
\\
&\mathrm{P_{II}} :\quad \frac{d^2\lambda}{d t^2}=
2\lambda^3+t\lambda+\alpha, 
\end{align*}
where $\lambda=\lambda(t)$ and  $\alpha$, $\beta$, $\gamma$, $\delta$ are 
complex parameters. We consider the generic  $\mathrm{P_{III}}$, 
that is, we assume $\gamma\delta\neq 0$. 
The third and second Painlev\'e equations are written as Hamiltonian systems 
\begin{align*}
\frac{d \lambda}{dt}=\frac{\partial H}{\partial \mu},\quad 
\frac{d \mu}{dt}=-\frac{\partial H}{\partial \lambda}
\end{align*}
with Hamiltonians:
\begin{align*}
& tH_{\mathrm{III}}=\lambda^2\mu^2+\left(2\theta_1\lambda+t\right)\mu-\theta_2\lambda
-\frac{\lambda^2}{4}+\theta_1^2,
\\
& H_{\mathrm{II}}=\frac{\mu^2}{2}-\left(\lambda^2+\frac{t}{2}\right)-\frac{\theta}{2}\lambda,
\end{align*}
where $(\alpha,\beta,\gamma,\delta)=(8\theta_1,4-8\theta_2,4,-4)$ for $\mathrm{P_{III}}$, 
and $\alpha=(\theta-1)/2$ for $\mathrm{P_{II}}$. The Hamiltonians satisfy the following 
differential equations:
\begin{align}
&\mathrm{P_{III}}:\quad (th'')^2-(4(h')^2-1)
(h-th')+4\theta_1\theta_2h'-\theta_1^2-\theta_2^2=0,
\label{eq_HIII_D}
\\
&\mathrm{P_{II}}:\quad (H'')^2-2H'(H-tH')+4(H')^3-\frac{\theta^2}{4}=0,  
\label{eq_HII_D}
\end{align}
where $h=tH$, and $'=d/dt$. 
Conversely,  a function satisfying the differential equation \eqref{eq_HIII_D} or 
\eqref{eq_HII_D} recovers the Painlev\'e function by 
\begin{align*}
&\mathrm{P_{III}} :\quad \lambda=-\frac{2th''+4\theta_1h'-2\theta_2}
{4(h')^2-1},
\\
&\mathrm{P_{II}} :\quad \lambda=\frac{2H''+\frac{\theta}{2}}{4H'}. 
\end{align*}
 We introduce the tau functions by $H_{\mathrm{J}}=(\log(\tau_{\mathrm{J}}(t))'$ 
 ($\mathrm{J}=\mathrm{III, II}$).

We expect that a ramified irregular conformal block of type $(0,0,1/2)$ in Example 
\ref{ex_001/2} is for $\mathrm{P_{III}}$. 
If we set 
\begin{equation*}
\Delta'=\frac{(\theta_1-\theta_2)^2}{4},\quad \beta=4\nu, \quad 
\Delta=\frac{(\theta_1+\theta_2)^2}{4},\quad c=1, 
\end{equation*}
then the coefficients of $z^{1/2}$, $z$ of 
the ramified irregular conformal block in Example \ref{ex_001/2} is equivalent to the coefficients  
$D_1(\nu)$, $D_2(\nu)$ appeared in (A.31) of \cite{BLMST}.  
We denote by $G(x)$ the Barnes $G$-function such that $G(x+1)=\Gamma(x)G(x)$, $G(1)=1$, 
where $\Gamma(x)$ is the Gamma function. 
In terms of ramified irregular conformal blocks, a series expansion of 
the tau function of the third Painlev\'e equation is given as follows. 
\begin{conj}
A series expansion of the tau function of the third Painlev\'e equation  at the 
irregular singular point $\infty$ is given by 
\begin{align*}
\tau(t)=&t^{-\theta_1\theta_2}e^{-t/2}\sum_{n\in\Z}
s^n 2^{-(\nu+n)^2}
G(1+\nu+n\pm(\theta_1+\theta_2)/2)\nonumber
\\
&\times \left\langle \frac{(\theta_1-\theta_2)^2}{4}\left|\ 
 \Phi^{(\theta_1+\theta_2)^2/4,4(\nu+n)}_{(1,0),(1,0)}
(t^{-1})\ \right|(1,0)\right\rangle. 
\end{align*}
Namely, $h=t (\log(\tau(t)))'$  satisfies the differential equation \eqref{eq_HIII_D}. 
\qed 
\end{conj}

\begin{re}
We observe that for positive integers $p,q$ and 
a ramified irregular conformal blocks of half rank with $c=1$, 
we have  
\begin{align*}
\beta_1^{p,q,i}=-2(p+q-2)+4(i-1)
\quad (i=1,\ldots, p+q-1). 
\end{align*}
Hence, the difference $\beta_1^{p,q,i+1}-\beta_1^{p,q,i}=4$, which 
is equal to the number $4$ appeared as the shift in the 
series expansion of the Painlev\'e $\mathrm{III}$ tau function above. 
In other words, the adjacent irregular conformal blocks are related by 
the screening operator. 
\qed  
\end{re}

We expect that a ramified irregular conformal block of type $(0,3/2)$ in Example 
\ref{ex_03/2} is for $\mathrm{P_{II}}$. 
If we set 
\begin{equation*}
\beta=\frac{4}{3}\nu, \quad 
\Delta=\frac{\theta^2}{4},\quad c=1, 
\end{equation*}
then the coefficients of $z^{3/2}$, $z^3$ of 
the ramified irregular conformal block in Example \ref{ex_03/2} is equivalent to the coefficients 
$D_1(\nu)$, $D_2(\nu)$ in (3.33)  of \cite{BLMST}.  
In terms of ramified irregular conformal blocks, a series expansion 
 of the tau function of the second Painlev\'e equation is given as follows. 
\begin{conj}
A series expansion of the tau function of the second Painlev\'e equation at the 
irregular singular point $\infty$ is given by 
\begin{align*}
\tau(t)=&t^{-\theta^2/2}\sum_{n\in\Z}
s^n(2\pi)^{-\nu-n}(4\sqrt{2})^{-(\nu+n)^2}e^{\pi \sqrt{-1}\nu^2/2}a^{-3(\nu+n)^2/2}
G(1+\nu+n\pm\theta/2)\nonumber
\\
&\times \left\langle 0\left|\ 
 \Phi^{\theta^2/4,4(\nu+n)/3}_{(0,1,0),(0,1,0)}
(at^{-1})\ \right|(0,1,0)\right\rangle, 
\end{align*}
where $a$ is a complex number satisfying $a^{3/2}=-\sqrt{-2}$. 
Namely, $H= (\log(\tau(t)))'$  satisfies the differential equation \eqref{eq_HII_D}. 
\qed 
\end{conj}

\textbf{Acknowledgments.}
The author would like to thank the referee for
 his/her valuable comments and suggestions. 
This work  
 was partially supported by JSPS KAKENHI Grant Number JP15K17560. 

%%%%%%%%%%%%%%%%%%%%%%%%%%%%%%%%%%%%%%%%%%%%%%%%%%%%

\end{document}